\documentclass[lettersize,journal]{IEEEtran}
\usepackage{amsmath,amsfonts,mathtools}
\usepackage{algorithmic,array}
\usepackage{textcomp, stfloats, url, verbatim}
\usepackage{graphicx,pgfplots,tikzscale,tikz,color,xcolor}
\usepackage{cite}
\usepackage[standard]{ntheorem}
\usepackage[caption=false,font=normalsize,labelfont=sf,textfont=sf]{subfig}
\newcommand{\R}{{\mathbb R}}

\newcommand{\mc}{\mathcal}
\DeclareMathOperator{\diag}{diag}
\DeclareMathOperator{\diff}{d}
\newcommand{\ddt}{\tfrac{\diff}{\diff \!t}}
\newtheorem{assumption}{Assumption}

\def\showpgfsquare{\tikz[baseline=-0.75ex]\node[black,mark size=0.5ex]{\pgfuseplotmark{square*}};}
\def\showpgftriangle{\tikz[baseline=-0.5ex]\node[black,mark size=0.65ex]{\pgfuseplotmark{triangle*}};}
\newcommand\drawline[1][black]{%
  \raisebox{2pt}{%
    \tikz \draw[#1,line cap=butt] (0pt,0pt) -- (10pt,0pt);%
  }%
}

\allowdisplaybreaks

\definecolor{blue}{rgb}{0.00000,0.44700,0.74100}%
\definecolor{orange1}{rgb}{0.85000,0.32500,0.09800}%
\definecolor{orange2}{HTML}{fcbe95}
\definecolor{yellow}{rgb}{0.92900,0.69400,0.12500}%
\definecolor{purple}{rgb}{0.49400,0.18400,0.55600}%
\definecolor{green}{rgb}{0.46600,0.67400,0.18800}%
\definecolor{lightblue}{rgb}{0.30100,0.74500,0.93300}%
\definecolor{magenta}{HTML}{CE2842}

\begin{document}

\title{Input-Output Specifications and Dynamic Droop Coefficients: Stability and Performance Conditions for Grid-Forming IBRs}

\author{Jennifer T. Bui, Dominic Gro\ss}



\maketitle

\begin{abstract}
This paper proposes dynamic stability and performance conditions for grid-connected inverter-based resources (IBRs).  To this end, we extend the notion of steady-state droop coefficients to dynamic droop coefficients to capture the small-signal dynamics of IBRs and synchronous generators (SGs). Notably, the dynamic droop coefficients can be obtained from input-output data collected at the unit's (e.g., IBR or SG) point of interconnection without requiring prior knowledge of IBR internals or controls structure. To obtain frequency stability conditions, this IBR model is combined with a lightweight dynamic transmission network model that accounts for uncertainty of line dynamics. The resulting stability conditions are highly scalable and, given a few key network parameters, can be verified at the unit level. To make the conditions practical and offer intuitive and illustrative interpretations, we map the frequency stability conditions to bounds on the Bode plot of the dynamic droop coefficient for two broad types of IBR responses. Moreover, our specifications on the dynamic droop coefficient (i) translate basic frequency control ancillary services into verifiable requirements, and (ii) provide insights into the much-debated question of how to certify an IBR as grid-forming (GFM). The results are illustrated using dynamic droop coefficients obtained using detailed simulations of GFM and GFL IBRs as well as SGs. 
\end{abstract}


\section{Introduction}

Electric power systems are undergoing an unprecedented transition toward large-scale integration of renewable generation and energy storage interfaced by power electronics. Replacing conventional synchronous generators (SGs) with inverter-based resources (IBRs) results in significantly different power system dynamics and challenges standard operating paradigms and controls on timescales from seasons to milliseconds. In the context of power system dynamics, the heterogeneous dynamics of renewables represent a substantial barrier to their large-scale deployment. Specifically, scalable and reliable operation of today's systems crucially relies on the homogeneous physics and controls of SGs~\cite{PM2020}. In contrast, the dynamics of IBRs vastly differ across technologies, resulting in interoperability concerns that are a major obstacle to replacing few large and approximately homogeneous SGs with many small and heterogeneous IBRs.

Typically, control strategies for IBRs are broadly categorized into (i) grid-following (GFL) controls that require a stable ac voltage at their point of interconnection (i.e., ensured by SGs), and (ii) grid-forming (GFM) controls that impose a stable and self-synchronizing ac voltage at their terminals. Hence, GFM IBRs are envisioned to be the cornerstone of future power systems~\cite{gfmkeyetal}. Prevalent GFM controls include droop control~\cite{CDA1993}, virtual synchronous machine (VSM) control~\cite{DSF2015}, and dispatchable virtual oscillator (dVOC) control~\cite{GCB+2019}. Notably, large-signal stability conditions for homogeneous IBRs have been established using quasi-steady-state models for droop control \cite{Schiffer+2014} and, more comprehensively, accounting for transmission dynamics and inner controls for dVOC \cite{Subotic+2021}.

Since droop, dVOC, and VSM controls can be tuned such that their reduced-order models coincide \cite[Fig.~7]{DG2023}, small-signal stability conditions for homogeneous IBR dynamics \cite{compensate} may extend to networks of heterogeneous IBRs. However, this approach does not establish formal interoperability guarantees for IBRs with heterogeneous dynamics and does not address concerns around interoperability of GFM IBRs with SGs and GFL IBRs. Crucially, operating emerging power systems requires GFM IBRs to replace functions of SGs while accommodating (legacy) GFL IBRs for renewables and reliably operating HVDC transmission~\cite{GSA+2021}. At the same time, a wide range of adverse interactions between SGs and GFM IBRs~\cite{CTG+2020}, SGs and GFM and GFL IBRs~\cite{MSA+2021}, and GFM and GFL IBRs~\cite{DWT+2023} have been reported, highlighting the challenges arising from complex interactions between layered control structures and network components~\cite{MSA+2021, Wan+2015}.

Well-known methods relying on eigenvalue sensitivity~\cite{eigen+2007, Marta+2020} and impedance models~\cite{LXL2018} can be used to study the small-signal stability of power systems. However, these methods typically require constructing a model of the entire system to identify oscillatory and unstable modes, resulting in scalability issues for systems with many generation units. While data-driven impedance modeling reduces modeling requirements, scalability challenges still remain for large networks. To overcome scalability challenges, decentralized methods have been proposed that, broadly speaking, use a coarse model of the transmission network to develop stability conditions that can be verified for individual units. For instance, \cite{compensate,GCB+2019} provides stability conditions for homogeneous IBRs that rule out adverse interactions with transmission line dynamics. Decentralized stability conditions are presented in~\cite{GainPhase} that compare the gain and phase of the unit admittance matrices against those of the network admittance. To evaluate this condition in mixed GFM / GFL systems, the overall power system is decomposed into two subsystems, one consisting of GFM IBRs and SGs and another one with only GFL IBRs, such that the former serves as an equivalent grid for the stability conditions on the latter. While significantly improving scalability, this approach nonetheless requires constructing partial models of the overall grid.

Besides ensuring stability, performance (e.g., frequency regulation) is of importance and typically enforced via standards and requirements for ancillary services. One of the earliest attempts to formalize performance specifications for GFM IBRs evaluated the frequency disturbance rejection capability of a unit~\cite{powertech}. In particular,~\cite{powertech} recovered the transfer function between frequency oscillations from an ac voltage source and the IBR frequency control reference to evaluate the IBR disturbance response.  Although the results relied on an internal signal that is not externally measurable, they provided insight into developing specifications based on input-output data via a frequency-gridding approach. By instead leveraging experimentally accessible signals such as phase angle, voltage magnitude, and power, this approach allows for studying subsynchronous oscillations attributed to the outer control loops of IBRs that largely determine the interactions between interconnected units. While impedance models can capture the fast dynamics of IBRs, they are only linear in the controller $dq$ frame coordinates corresponding to the inner control loops.

The main contributions of this work is a simple data-enabled modeling and analysis framework for evaluating local frequency control performance and certifying small-signal frequency stability of networks of heterogeneous units (e.g., IBRs and SGs). To this end, we leverage data-enabled input-output modeling to extend the definition of steady-state droop coefficients to dynamic droop coefficients that can intuitively be represented through Bode plots. Next, we present a method for recovering input-output models using simulations or experiments that does not require knowledge of the internal hardware and controls structure.  We formulate a decentralized stability condition that provides bounds on the unit dynamics by extending the result in~\cite{PM+2019} to include network circuit dynamics.  Our conditions also account for the impact of modeling uncertainty in the network circuit dynamics (i.e., transmission lines).  In addition, we propose unit-level performance specifications that translate common frequency support functions into quantifiable conditions on the dynamic droop coefficients. We apply our conditions to a variety of generation technologies and illustrate the use of the proposed framework to (i) characterize the impact of the minimum network inductance on stability, (ii) delineate between GFM and GFL IBRs, and (iii) illustrate the impact of network circuit dynamics.

\section{Motivation and Problem Setup\label{sec:motivation}}
\begin{figure}[b]
  \centering
  \includegraphics[width = 3.5in]{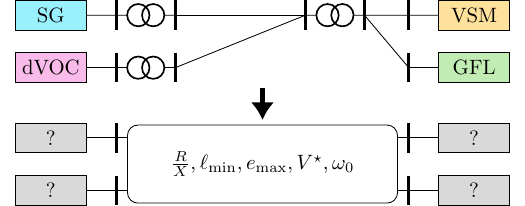}
  \vspace*{3mm}
  \caption{Abstraction of an original network (top) through an analytical model (bottom) using the minimum line inductance $\ell_{\min}$, resistance-to-reactance ratio $\frac{R}{X}$, maximum number of lines $e_{\max}$ connected to a bus, and nominal voltage magnitude and frequency $V^\star$ and $\omega_0$.
  \label{fig:abstracted-network}}
\end{figure}

To ensure frequency stability of large-scale systems of heterogeneous units (i.e., IBRs and SGs), we focus on developing scalable specifications on the individual unit dynamics. Broadly speaking, SGs and GFM IBRs synchronize with other units in the grid through their local control response. In this sense, restrictions on the unit dynamics also need to account for the grid characteristics. To this end, we consider (i) a network model that, through a few key network parameters, abstractly captures how the network circuit (e.g., topology, line parameters) impacts stability, and (ii) technology-agnostic unit-level specifications that each unit should meet for a given set of network parameters. The key network parameters are the $R/X$ ratio, the smallest outgoing inductance $\ell_{\min}$ between units in the Kron-reduced network (i.e., when eliminating buses without SGs/IBRs), and the largest number of outgoing lines $e_{\max}$. These parameters, in abstraction, model connectivity and coupling strength (i.e., short-circuit ratio) for the purpose of developing analytical stability certificates (see Fig.~\ref{fig:abstracted-network}).

In practice, the details of the internal control or hardware implementation of a unit may not be known. Thus, we impose specifications on the input-output dynamics of units represented in the frequency domain. In particular, the dynamic response of bus voltages (i.e., frequency and magnitude) to network power injections (i.e., active and reactive power) is captured at discrete frequencies $\omega_p$ (see Fig.~\ref{fig:decentralized-condition}) using data collected in simulation or experiments.

\begin{figure}[t!!]
  \centering
  \includegraphics[width = 0.8\columnwidth]{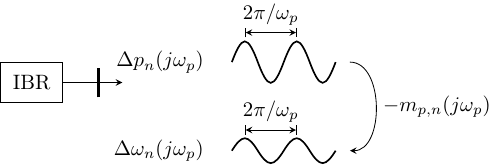}
  \caption{The gain and phase shift of the frequency deviation $\Delta \omega_n$ at bus $n$ relative to the active power $\Delta p_n$ at frequency $\omega_p$ is captured through the dynamic droop coefficient $m_{p,n}(j \omega_p) \in \mathbb{C}$.\label{fig:decentralized-condition}}
\end{figure}

This frequency domain model of each unit can then be certified against the stability conditions and performance specifications in isolation to (i) certify small-signal frequency stability for all network circuits with a fixed resistance-to-reactance ratio $\R/X$ that are within the bounds parameterized by $\ell_{\min}$ and $e_{\max}$, and (ii)  understand the impact of the network parameters on stability.

While the system may still be stable despite not meeting the decentralized stability conditions, the lack of a-priori guarantees that are independent of the precise network parameters results in scalability and robustness concerns. In other words, the system may be unstable or certifying stability may require exact system (i.e., network and all units) models that are unavailable or may change in real-time. This second approach is not scalable nor amenable to developing grid codes. In contrast, decentralized stability certificates can directly inform standards to ensure stable interconnection for a wide range of system configurations based on unit-level stability certificates.

\section{Power System Modeling \label{sec:model}}

\subsection{Network topology}
We model the power system as an undirected graph $\mc G=(\mc N,\mc E)$ with nodes $\mc N$ corresponding to $|\mc N|$ buses and edges $\mc E \subseteq \mc N \times \mc N$ corresponding to $|\mc E|$ transmission lines. Here, $|\mc X|$ denotes the cardinality of a set $\mc X \subset \mathbb{N}$. To every bus $n \in \mc N$, we associate a voltage phase angle $\theta_n$, frequency $\omega_n=\ddt \theta_n$, voltage magnitude $V_n$, and active and reactive power injections $p_n$ and $q_n$. To every line $m \in \{1,\ldots,|\mc E|\}$ connecting buses $(n,k) \in \mc E$ we associate an active and reactive power $\bar{p}_{n,k}$ and $\bar{q}_{n,k}$. Thus, the bus power injections are given by $p_n= \sum_{k: (n,k) \in \mathcal{E}}  \bar{p}_{n,k}$ and $q_n= \sum_{k: (n,k) \in \mathcal{E}} \bar{q}_{n,k}$. In the remainder, the network topology is encoded by an oriented incidence matrix $B \in \{-1,0,1\}^{|\mc N| \times |\mc E|}$ (see~\cite{FB-LNS}). 


\subsection{Dynamic droop model of IBR and SG dynamics}\label{subsec:dyndroop}
We model the dynamics of IBRs and SGs at bus $n \in \mc N$ through the transfer function model
\begin{align}\label{eq:busdyn}
  \begin{bmatrix} \Delta \omega_n (s) \\ \Delta V_n(s) \end{bmatrix} =
  -\underbrace{\begin{bmatrix}
    m_{p,n}(s) & \zeta_{q,n}(s) \\ \zeta_{p,n}(s) & m_{q,n}(s)
  \end{bmatrix}}_{\eqqcolon M(s)}
  \begin{bmatrix} \Delta p_n(s) \\ \Delta q_n(s) \end{bmatrix},
\end{align}
where $\Delta(\cdot)$ denotes deviation of a signal from its set point. We emphasize that \eqref{eq:busdyn} simply encodes the dynamic relationship between the four signals and the model structure is not necessarily tied to the inputs or outputs of the IBR control. For instance, a GFL IBR may control the power injection as a function of frequency and voltage magnitude, while a GFM IBR typically controls frequency and voltage magnitude as a function of power. The dynamic droop model \eqref{eq:busdyn} is complementary to well-studied impedance models~\cite{eigen+2007,LXL2018} and extends the common definition of steady-state droop coefficients to dynamic droop coefficients that fully characterize the IBR small-signal response below the nominal line frequency (e.g., subsynchronous oscillations).

We refer to the diagonal elements $m_{p,n}(s)$ and $m_{q,n}(s)$ of $M(s)$ as the dynamic droop coefficients that reflect the $P \text{-} f$ and $Q \text{-} V$ relationships, respectively. Notably, the commonly used steady-state droop coefficients $m_{p,0} \in \mathbb{{R}}_{>0}$ and $m_{q,0} \in \mathbb{{R}}_{>0}$ can be recovered as $\lim_{\omega_p \rightarrow 0} m_{p,n}(j\omega_p)=m_{p,0}$ and $\lim_{\omega_p \rightarrow 0} m_{q,n}(0)=m_{q,0}$. The gains $|m_{p,n}(j\omega_p)|$ and $|m_{q,n}(j\omega_p)|$ specify the change in magnitude of frequency and voltage magnitude in response to active and reactive power oscillations of frequency $\omega_p$. Moreover, $\angle m_{p,n}(j\omega f_p)$ and $\angle m_{q,n}(j\omega_p)$ captures the corresponding phase shift.

Examples of different $P \text{-} f$ droop behaviors in time-domain and their corresponding dynamic droop coefficients $m_{p,n}(j \omega_p) \in \mathbb{C}$ are illustrated in Fig.~\ref{fig:dyn_droop}. For instance, a $5\%$ per unit droop $m_{p,n}(j \omega_p)=0.05$~pu results in a $0.005$~pu bus frequency deviation for a $0.1$~pu active power oscillation at $\omega_p$. Meanwhile, $m_{p,n}(j\omega_p)=0.03$~pu results in a smaller frequency deviation, indicating increased damping. For both cases, $m_{p,n}(j\omega_p)$ has zero phase shift relative to conventional $P\text{-}f$ droop constants. In contrast, $m_{p,n}(j\omega_p)=0.03 e^{j\pi/2}$ results in a $\angle m_{p,n}(j\omega_p) = 90^\circ$ phase shift relative to the standard $P \text{-} f$ droop, i.e., droop with an incorrect sign for half the oscillation. If the phase exceeds $\pm 90^\circ$, then the IBR droop has incorrect sign for more than half the oscillation resulting in \emph{negative damping} on average over a full cycle of the oscillation. Meanwhile, the off-diagonal elements, $\zeta_{p,n}(s)$ and $\zeta_{q,n}(s)$, model cross-coupling (i.e., $P\text{-}V$ and $Q\text{-}f$).
\begin{figure}[h!!!!]
  \centering
  \includegraphics[width=\columnwidth]{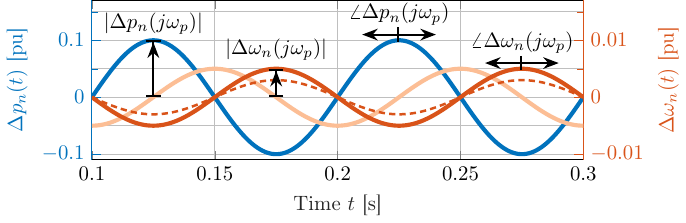}
  \caption{
    Change in instantaneous active power injection $\Delta p_n(t)$ (\drawline[blue, line width=1.3pt]) for an $\omega_p = 2\pi10$~Hz oscillation along with the change in frequency $\Delta \omega_n(t)$ corresponding to three different active power dynamic droop coefficients $m_{p,n}(j\omega_p)=0.05$ pu (\drawline[orange1, line width=1.3pt]), $m_{p,n}(j\omega_p)=0.03$ pu (\drawline[orange1, dashed, line width=1.1pt]), and $m_{p,n}(j\omega_p)=0.05e^{j\pi/2}$ pu (\drawline[orange2, line width=1.3pt]).
    \label{fig:dyn_droop}}
\end{figure}

\subsection{Transmission Line}
Analogously to the unit model in \eqref{eq:busdyn}, the power flowing into a transmission line $m \in \{1,\ldots,|\mc E|\}$ connecting buses $(n,k) \in \mc E$ with nominal bus voltages $V^\star_{n}, V^\star_{k} \in\mathbb{R}_{>0}$ and resistance-inductance ratio $\rho_m=\frac{r_m}{l_m} \in \mathbb{R}_{>0}$ is modeled by
\begin{align}\label{eq:netdyn}
  \!\!\begin{bmatrix} \Delta \bar{p}_{n,k}(s) \\ \Delta  \bar{q}_{n,k}(s) \end{bmatrix} =
   \kappa_m \mu_m(s)
   \begin{bmatrix}
     1& \frac{\rho_m+s}{\omega_0 V^\star_{k_m}} \\ \frac{\rho_m+s}{\omega_0} & \frac{1}{V^\star_{k_m}}
   \end{bmatrix}
  \begin{bmatrix} \Delta \theta_{n,k}(s)\\ \Delta  V_{n,k}(s) \end{bmatrix}\!,\!
\end{align}
with voltage differences $\Delta  \theta_{n,k} \coloneqq \Delta \theta_n-\Delta \theta_k$, $\Delta  V_{n,k} \coloneqq \Delta V_n- \Delta  V_k$, nominal  frequency $\omega_0 = 2\pi f_0$, and 
\begin{align}
  \kappa_m \coloneqq \frac{\omega_0 V^\star_n V^\star_k}{\ell_m},  \quad \mu_m(s) \coloneqq \frac{1}{s^2+2 \rho_m s+\omega_0^2+\rho^2_m}.
  \label{eq:line-dyn}
\end{align}
%
This model coincides with the common DC power flow model for low frequencies, but also captures network circuit dynamics that play a significant role in stability analysis of IBR-dominated power systems~\cite{GCB+2019,compensate}. For analysis purposes, we define the DC gain $\mu_0 = \lim_{f_p \to 0} |\mu(j 2\pi f_p )|$ and peak gain $\hat{\mu} = |\mu(j 2\pi f_0)|$. Moreover, we define $\epsilon_l \in \mathbb{R}_{>0}$, $\delta_l \in \mathbb{R}_{>0}$, and $\nu_l \in \mathbb{R}_{>0}$ and frequencies $f_{\epsilon}$, $f_\delta$, and $f_{\nu}$ such that
\begin{subequations}\label{eq:mubound}
  \begin{align}
  |\mu(j 2\pi f_p)| &\leq \mu_0 + \epsilon_{l}, \quad &\forall f_p &\leq f_{\epsilon},\label{eq:mubound:feps}\\
  \angle \mu(j2\pi f_p) &\geq -\delta_l, \quad &\forall f_p &\leq f_{\delta},\label{eq:mubound:fdelta}\\
  \angle \mu(j2\pi f_p) &\leq \nu_l -\pi, \quad &\forall f_p &\geq f_{\nu}.
\end{align}
\end{subequations}

While transmission lines also exhibit significant higher-order circuit dynamics~\cite{D1985,MGT1999}, it is often difficult to obtain accurate models and parameters beyond the model \eqref{eq:netdyn}. Thus, to account for higher-order circuit dynamics neglected in \eqref{eq:line-dyn}, we consider a line model with additive uncertainty~\cite[Ch.~7.4.2]{skogestad1996multivariable}. To this end, consider the frequency dependent weight $W_m(s) \coloneqq \beta_m   \tfrac{\omega_{\Delta,m} s}{s+\omega_{\Delta,m}}$ and the norm bounded uncertainty $||\Delta_m(s)||_{\infty} \leq 1$. This results in
\begin{align}\label{eq:muaddunc}
 \mu_{\Delta,m}(s) & = \mu_m(s) +  W_m(s) \Delta_m(s),
\end{align}
and it follows that $|\mu_{\Delta,m}(j\omega_p)-\mu_m(j\omega_p)| \leq |W_m(j\omega_p)|$. The normalized nominal line model $\mu_m(s)$, weight $W_m(s)$, and maximum gain of $| \mu_{\Delta,m}(s)|$ are illustrated in Fig.~\ref{fig:line-uncertainty}. Notably, the uncertainty is approaching its maximum at the corner frequency $\omega_{\Delta,m}$. Thus, a smaller corner frequency will decrease the maximum gain of the uncertainty. In that case, $\beta_m$ can be increased to model larger uncertainty. 
 
\begin{figure}[b!!!]
  \centering
  \includegraphics[width=\columnwidth]{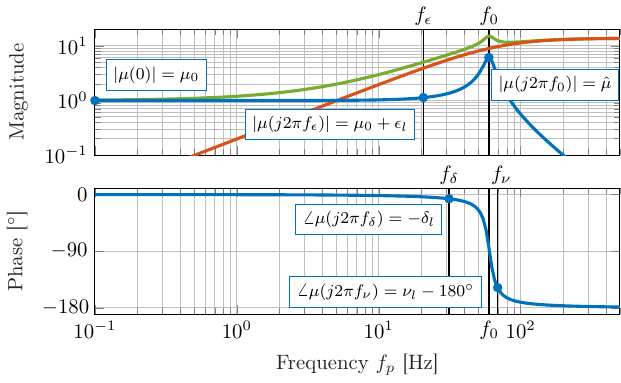}
  \caption{
    Normalized dynamic line model $\mu(j2\pi f_p)$ (\drawline[blue, line width=1.3pt]) for $\rho = 0.0817$ pu, weight $|W(j2\pi f_p)|$ (\drawline[orange1, line width=1.3pt]), and maximum gain of $\mu_{\Delta}(j2\pi f_p)$ (\drawline[green, line width=1.3pt]).
  \label{fig:line-uncertainty}}
\end{figure}

\subsection{Multi-IBR/multi-SG system frequency dynamics}
The following assumption is often justified for transmission systems and used to simplify the frequency dynamics.
\begin{assumption}{\bf(Decoupled active and reactive power)}\label{ass.coupling}
  We assume that active and reactive power are decoupled, i.e., the off-diagonal elements in \eqref{eq:netdyn} are negligible.
\end{assumption}
We use $\psi_n \in \mathbb{R}_{>0}$ to denote the unit rating at every bus $n \in \mc N$ relative to the system base and the matrix $\Psi \coloneqq \diag\{\psi_n\}_{n=1}^{|\mc N|}$ collecting unit ratings. We define the matrix $G_{\omega,p}(s) \coloneqq\diag\{m_{p,n}(s)\}_{n=1}^{|\mc N|}$ of bus transfer functions with $m_{p,n}(s)$ normalized by $\psi_n$. The vector of bus power injections $\Delta p \in \mathbb{R}^{|\mc N|}$ is then given by $\Delta p = B G_{\omega,p}(s) B^\mathsf{T} \Delta \theta$ and $\delta_p(s)$ denotes the vector of load perturbations.  Finally, we define the matrix of transmission line transfer functions $G_{p,\theta}(s) \coloneqq \diag\{\kappa_m \mu_m(s)\}_{m=1}^{|\mc E|}$.  The small-signal frequency dynamics from are shown in Fig.~\ref{fig:smallsignal}. We emphasize that this model can capture the frequency dynamics of a wide range of units from GFM and GFL IBRs to SGs. 

\begin{figure}[h!!]
  \centering
  \includegraphics[width=0.7\columnwidth]{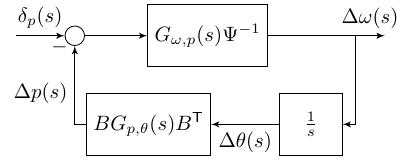}
  \caption{Small-signal frequency dynamics of the multi-IBR/multi-SG system.
  \label{fig:smallsignal}}
\end{figure}

To simplify the analysis, we make the following assumption that is typically justified for lines at the same voltage level.
\begin{assumption}{\bf(Uniform resistance-inductance ratio)}\label{ass.constant.ratio}
The resistance-inductance ratio is identical for all lines, i.e., $\rho_m=\rho$ for all $m \in \{1,\ldots,|\mc E|\}$.
\end{assumption}
Under Assumption~\ref{ass.constant.ratio}, $\mu(s)=\mu_m(s)$ for all $m \in \{1,\ldots,|\mc E|\}$. Letting $K \coloneqq \diag\{\kappa_m\}_{m=1}^{|\mc E|}$, the network circuit dynamics simplify to $B G_{p,\theta}(s) B^\mathsf{T}=\mu(s) B K B^\mathsf{T}$, where $L \coloneqq B K B^\mathsf{T}$ is the graph Laplacian of the network. We emphasize that the inductance of each line $m$ is preserved by $\kappa_m$ in \eqref{eq:line-dyn}.

\section{Data-Driven Modeling \label{sec:validation}}
Conceptually, well-known system identification methods could be used to obtain the bus transfer functions \eqref{eq:busdyn}. However, these typically fit a model with prescribed numbers of poles and zeros to the data, which are typically unknown in this application. Instead, we connect the unit under test (UUT) to a controlled ac voltage source (see Fig.~\ref{fig:experiment}) to obtain its response at discrete perturbation frequencies $\omega_p\in \mathbb{R}_{>0}$.

\subsection{Data-driven dynamic droop model}
To recover the dynamic droop model \eqref{eq:busdyn}, we follow a three-step procedure.
\begin{figure}[t!!]
  \centering
  \includegraphics[width=3.5in]{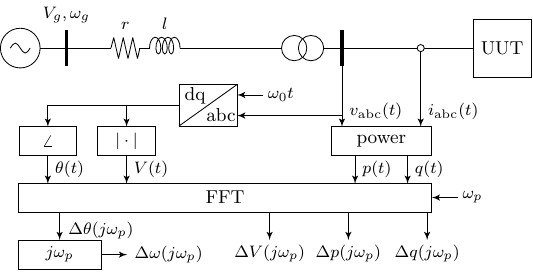}
  \caption{Two-bus system and signal processing for capturing frequency response models.}
  \label{fig:experiment}
\end{figure}
\subsubsection{Probing signal} 
To excite the UUT, the ac voltage source is imposing the voltage $V_{g,\text{abc}}(t) = V_g(t) \cos(\theta_g(t)+\bar{\theta}) \in \mathbb{R}^3$, where $\bar{\theta}=(0,-\frac{2}{3}\pi,\frac{2}{3}\pi)$ and $V_g(t)$ and $\theta_g(t) = \int_{0}^\tau \omega_g(\tau) \diff \tau$ denote the voltage magnitude and phase given by
\begin{subequations}
  \begin{align}
    V_g(t) &= V^\star + A_V \sin(\omega_p t),\\
    \omega_g(t) &= \omega_0 + A_\omega \sin(\omega_p t).
  \end{align}
\end{subequations}
Here $V^\star \in \mathbb{R}_{>0}$ denotes the nominal voltage magnitude, while $A_V$ and $A_\omega$ denote the amplitude of the perturbation with frequency $\omega_p$.  The amplitudes are chosen such that the resulting oscillations in the measured signals are large enough to sufficiently excite the UUT dynamics but small enough to remain in the realm of small-signal modeling (e.g., do not activate limiters).  To recover the dynamic droop model \eqref{eq:busdyn}, for each discrete frequency $\omega_p$, data is collected by individually perturbing the voltage magnitude (i.e., $A_V \in \mathbb{R}_{>0}$ and $A_\omega = 0$) and frequency (i.e., $A_V=0$ and $A_\omega \in \mathbb{R}_{>0}$).

\subsubsection{Processing of UUT input/output signals}
The three-phase voltage $v_{\text{abc}}$ and current $i_{\text{abc}}$ at the UUT terminal (see Fig.~\ref{fig:experiment}) are used to compute the active and reactive power $p$ and $q$. Moreover, the magnitude $V$ and phase angle $\theta$ of $v_{\text{abc}}$ are computed in a synchronous reference frame rotating at the nominal frequency $\omega_0$. Next, a Fourier analysis is performed to obtain the amplitude and phase of the oscillations of the four signals at the frequency $\omega_p$. This allows for representing the magnitude and phase of the resulting oscillations as complex numbers $\Delta \theta(j\omega_p)$, $\Delta V(j\omega_p)$, $\Delta p(j\omega_p)$, and $\Delta q(j\omega_p)$. Finally, we compute frequency $\Delta {\omega}(j\omega_p) = j\omega_p \Delta {\theta}(j\omega_p)$. Crucially, we only require measurements at the UUT terminal and do not require any knowledge of internal signals, controls, or the hardware topology of the UUT.

\subsubsection{Recovering the frequency response model}
The results of the experiments at every frequency are collected in vectors
\begin{align*}
  Y(j \omega_p) &=     \begin{bmatrix}
    \Delta{\omega}_{\omega}(j\omega_p) & \Delta{\omega}_{V}(j\omega_p) \\ \Delta{V}_{\omega}(j\omega_p) & \Delta{V}_{V}(j\omega_p)
  \end{bmatrix},\\ U(j \omega_p)&=     \begin{bmatrix}
    \Delta{p}_{\omega}(j\omega_p) & \Delta{p}_{V}(j\omega_p) \\ \Delta{q}_{\omega}(j\omega_p) & \Delta{q}_V(j\omega_p)
  \end{bmatrix},
\end{align*}
where the columns correspond to perturbations of the ac source frequency and voltage, respectively. Then, at any given frequency $\omega_p$, the dynamic droop model \eqref{eq:busdyn} is recovered by noting that $M (j \omega_p)= -Y(j \omega_p) U^{-1}(j \omega_p)$.

\subsection{Examples: SG and IBRs}\label{subsec:example}
The setup shown in Fig.~\ref{fig:experiment} has been implemented in an EMT simulation environment in MATLAB/Simulink. The Bode plot of the dynamic droop coefficient $m_{p}(j2\pi f_p)$ for a sixth-order SG model with turbine/governor system, automatic voltage regulator, and power system stabilizer is shown in Fig.~\ref{fig:dynamics-results}. The DC gain of $|m_{p}(j2\pi f_p)|$ for the SG matches its droop setting of $0.05$~pu. As $f_p$ increases, the limited bandwidth of the turbine/governor system first results in an increase of $|m_{p}(j2\pi f_p)|$ and then the SG's inertia response results in a decrease of $|m_{p}(j2\pi f_p)|$. Beyond approximately $2~$Hz, the dynamics of the machine stator dominate the response resulting in an increase of $|m_{p}(j2\pi f_p)|$. Finally, $\angle m_{p}(j2\pi f_p)$ is mostly contained within $\pm90^\circ$, i.e., the $P\text{-}f$ dynamics of the SG are largely dissipative in the frequency range of interest.

Moreover, three GFM IBRs and two GFL IBRs using a synchronous reference frame (SRF) phase-locked loop (PLL) are studied. Fig.~\ref{fig:dynamics-results} presents the Bode plots of $m_{p}(j2\pi f_p)$ for the five IBR models configured for $5\%$ steady-state $P\text{-}f$ droop.
\begin{figure}[t!]
    \centering
    \includegraphics[width = \columnwidth]{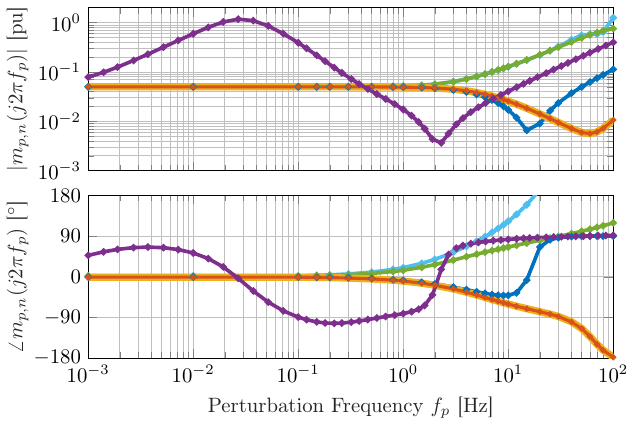}
    \caption{Recovered dynamic droop coefficients for 
    GFM droop without inner-loops (\drawline[blue, line width=1.3pt]), 
    GFM droop with inner-loops (\drawline[orange1, line width=1.3pt]),
    GFM dVOC control (\drawline[yellow, line width=1.3pt]),
    a SG (\drawline[purple, line width=1.3pt]),
    GFL SRF-PLL (\drawline[green, line width=1.3pt]),
    GFL SRF-PLL with delay (\drawline[lightblue, line width=1.3pt]).
    \label{fig:dynamics-results}} 
\end{figure}
Firstly, it can be observed from the Bode magnitude plot that all IBRs provide the desired steady-state droop of $0.05~$pu. Differences between GFM IBRs and GFL IBRs appear after approximately $2~$Hz. In particular, all three GFM IBRs exhibit increased frequency damping (i.e., decreasing gain) while the GFL IBRs exhibit decreasing frequency damping (i.e., increasing gain). The responses of the droop and dVOC GFM IBRs with dual-loop current and voltage control~\cite[Fig.~4b]{acpower} have identical responses and only exhibit increasing gain past the line frequency (i.e., $60$~Hz) when the inner loops no longer fully compensate their LC filter dynamics. In contrast, GFM droop without inner controls, exhibits increasing gain beyond $5$~Hz when its LC filter dynamics dominate the response. In other words, GFM IBRs with inner control loops can maintain the desired gain roll-off to higher frequencies.

\begin{figure*}[bh!!!]
    \centering
    \includegraphics[width = 1\textwidth]{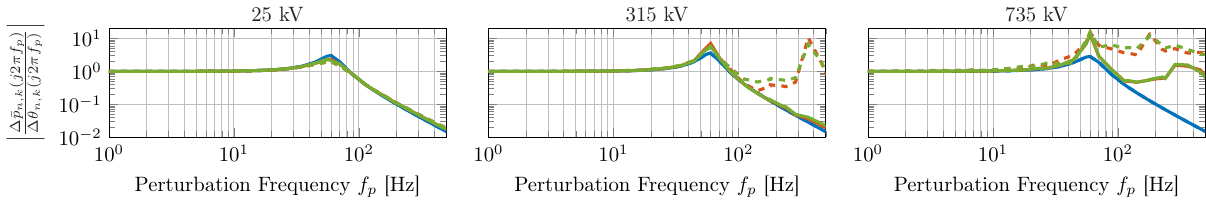}
    \caption{Recovered dynamics for transmission lines based on lumped-parameter (\drawline[blue, line width=1.3pt]),
  distributed-parameter (\drawline[orange1, line width=1.3pt]),
  frequency-dependent-parameter (\drawline[green, line width=1.3pt]) models, in which results for short lines are in solid and long lines are dashed.
    \label{fig:lines-diagonals}} 
\end{figure*}

However, the SG and GFM IBR without inner loops largely stay within $\pm90^\circ$, maintaining dissipative responses that avoid negative damping at the expense of reduced high-frequency oscillation damping (i.e., larger gain). In contrast, the GFM IBRs with inner loops and GFL IBRs exhibit phase shifts larger than $\pm 90^\circ$ beyond $30~$Hz (GFM), $20$~Hz (GFL), and $6~$Hz (GFL with delay) resulting in negative damping (see Sec.~\ref{subsec:dyndroop}). Crucially, the negative damping of GFM IBRs with inner loops beyond $30$~Hz is somewhat mitigated by the low gain, while the GFL IBRs exhibit negative damping with high gain beyond approximately $10$~Hz. Finally, the GFL IBR with delay exhibits droop with incorrect sign (i.e., $180^\circ$ phase shift) at approximately $19$~Hz and has been associated with a $18-20$~Hz oscillation events~\cite{DWT+2023}.
 
\subsection{Example: Transmission lines}\label{subsec:exampleline}
The data-driven approach can also be used to obtain Bode plots for transmission line models. To this end, the transfer function from angle difference $\Delta \theta_{n,k}$ to active power $\Delta \bar{p}_{n,k}(s)$ has been obtained in EMT simulations for line models of various fidelity. These results can be used to calibrate the weight $W_m(s)$ used to model higher-order dynamics beyond the reduced-order model $\mu(j\omega_p)$. Results for lumped-parameter, distributed-parameter, and frequency-dependent-parameter models are shown in Fig.~\ref{fig:lines-diagonals} for short and long lines at three voltage levels. Across all three voltage levels, it can be seen that the distributed-parameter and frequency-dependent-parameter models capture high-order resonances not captured by the lumped-parameter model (i.e., $\mu(s)$). As the voltage level and/or line length increases, larger resonances start to appear at lower perturbation frequencies. Moreover, the results indicate that the second-order dynamic line model \eqref{eq:line-dyn} captures dynamics reasonably well for high-voltage transmission and short lines while the frequency-dependent additive uncertainty in \eqref{eq:muaddunc} can be used to capture higher-order dynamics for long lines and extra-high voltage transmission.

\section{Stability under interconnection\label{sec:stability}}
\subsection{Decentralized Stability Certificates}
To ensure scalability, we focus on decentralized stability conditions that can be verified for every unit separately while ensuring stability under interconnection through a network parameretrized using a few key parameters. To state the decentralized stability condition, we define the constant $\gamma_n \coloneqq 2 \sum_{(n,k) \in \mc E} \frac{\omega_0}{\ell_{n,k}} V_n V_k$, where $\ell_{n,k}\in\mathbb{R}_{>0}$ denotes the inductance of the line connecting buses $n$ and $k$ that captures the coupling strength between bus $n$ and the remaining buses. Notably, $\gamma_n$ can be bounded by $\bar{\gamma} \coloneqq 2  \frac{e_{\max}}{\ell_{\min}} \omega_0 V^2_{\max}$, where $e_{\max} \in \mathbb{N}$ is an upper bound on the number of outgoing edges of all nodes, $\ell_{\min} \in \mathbb{R}_{>0}$ is a lower bound on the line impedances, and $V_{\max}$ is the maximum voltage magnitude. 
\begin{theorem}{\bf(Decentralized stability condition)}\label{th.stability}
    Assume that all poles of $\mu(j\omega_p)m_{p,n}(j\omega_p)$ are in the open left half-plane for all $n \in \mc N$. The overall frequency dynamics shown in Fig.~\ref{fig:smallsignalre} are asymptotically stable if there exists $\alpha \in [0,\pi/2)$ such that 
    \[ \operatorname{Re}\left\{e^{j(\alpha-\frac{\pi}{2})}(e^{j\frac{\pi}{2}}+\tfrac{\gamma_n}{\psi_n \omega_p} \mu(j\omega_p) m_{p,n}(j\omega_p))\right\}>0\]
    holds for all $n \in \mc N$ and $\omega_p \in \mathbb{R} \cup \{\infty\}$.
\end{theorem}
\begin{figure}[t]
  \centering
  \includegraphics[width=0.7\columnwidth]{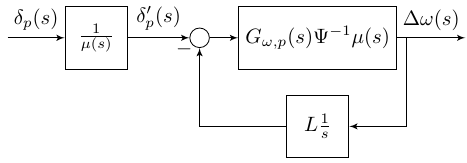}
  \caption{Rearranged small-signal frequency dynamics for a system with constant resistance-inductance ratio. \label{fig:smallsignalre}}
\end{figure} 

The proof follows by applying elementary block operations to obtain the block diagram shown in Fig.~\ref{fig:smallsignalre}. Because $\mu(s)^{-1}$ is asymptotically stable, it remains to show that the closed loop from $\delta^\prime_p(s)$ to $\Delta \omega(s)$ is stable. This follows from~\cite[Th.~1, Lem.~5, ]{PM+2019}, scaling $L$ as in~\cite[Fig.~4]{PM+2019}, noting that $\bar{\gamma}$ is an upper bound on the largest eigenvalue of the graph Laplacian $L$, and using the half-plane constraint in~\cite[Sec.~III-C]{PM+2019}. Notably, using $\tfrac{1}{j\omega_p} = \tfrac{1}{\omega_p}e^{-j\frac{\pi}{2}}$ in \cite[Lem.~5]{PM+2019} results in $\operatorname{Re}\left\{e^{j\alpha}(1+ e^{-j\tfrac{1}{2}\pi}\tfrac{\gamma_n}{\psi_n \omega_p} \mu(j\omega_p) m_{p,n}(j\omega_p))\right\}>0$. Factoring out $e^{-j\tfrac{1}{2}\pi}$ rotates the half-plane that defines the stable region. 


Theorem~\ref{th.stability} requires that the Nyquist plot of the dynamics 
\begin{equation} \label{eqn:interconnection}
  h_n(j\omega_p) = \tfrac{\gamma_n}{\psi_n \omega_p} \mu(j\omega_p) m_{p,n}(j\omega_p),
\end{equation}
for each bus $n$ are contained within a half-plane going through $(0,-1)$ with angle $\alpha$. We emphasize that, if $\gamma_n$ is not known for a given bus $n$, it can be replaced with the upper bound $\bar{\gamma}$ for the purpose of certifying stability.
\begin{figure}[b]
  \centering
  \includegraphics[width=\columnwidth]{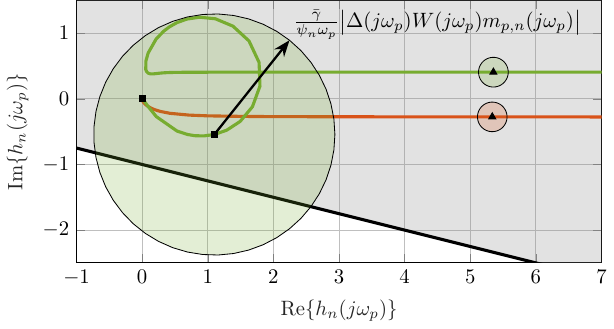}
  \caption{
    The Nyquist plot of $h_n(j\omega_p)$ for a GFM IBR with inner-loops (\drawline[orange1, line width=1.3pt]) and 
    GFL IBR (\drawline[green, line width=1.3pt]).  Uncertainty arising from line dynamics is shown for $f_p = 65$~Hz (\protect\showpgfsquare) and $f_p = 0.3$~Hz (\protect\showpgftriangle). The network parameters are $\bar{\gamma}=20.3$, $\ell_{\min}=0.197$~p.u. 
    , and $\psi=0.1$. 
    \label{fig:stability-nyquist-1}
  }
\end{figure}

Fig.~\ref{fig:stability-nyquist-1} illustrates the stability condition defined by Theorem~\ref{th.stability} along with the bus dynamics of a GFM and GFL IBR, which were obtained from the data shown in Fig.~\ref{fig:dynamics-results}.  Since the same half-plane bounds both Nyquist plots, the decentralized stability condition is satisfied for each connected unit and their stable interconnection on the same network is guaranteed.  

However, when considering uncertainty in the network circuit dynamics, the GFL IBR no longer satisfies the stability conditions of Theorem~\ref{th.stability}. In particular, accounting for uncertainty in the network circuit dynamics corresponds to checking the decentralized stability condition for the transfer function
  \begin{align*}
 \frac{\gamma_n}{\psi_n \omega_p} \big[\mu(j\omega_p) m_{p,n}(j\omega_p) + \Delta(j\omega_p)W(j\omega_p)m_{p,n}(j\omega_p)\big].
  \end{align*}
Notably, a circle with radius $\tfrac{\gamma_n}{\psi_n \omega_p} \big| W(j\omega_p)m_{p,n}(j\omega_p)\big|$ can be added to the Nyquist plot of the nominal transfer function for every $\omega_p$, as illustrated in Fig.~\ref{fig:stability-nyquist-1}.  It can be seen that the GFM IBR is significantly less sensitive to uncertainty in the line dynamics. For higher frequencies, the uncertainty $\big|W(j\omega_p)\big|$ is constant while the magnitude of the dynamic droop coefficient $|m_{p,n}(j\omega_p)|$ and $\frac{\gamma_n}{\psi_n\omega_p}$ are decreasing. Thus, the uncertainty does not significantly impact the GFM IBR.  For the GFL IBR response, the dynamic droop coefficient increases in magnitude as $\omega_p$ increases and, if $|m_{p,n}(j\omega_p)|$ increases faster than $|j\omega_p|$, the uncertainty grows unbounded as $\omega_p$ increases. In our example, $|m_{p,n}(j\omega_p)|$ grows at approximately the same rate as $|j\omega_p|$ resulting in significantly reduced robustness to unmodeled dynamics. Furthermore, the network parameter $\gamma_n \leq \bar{\gamma}$ and rating $\psi_n$ will also impact both the nominal Nyquist plot and circle modeling uncertainty.

\subsection{Interpretation of $\alpha$: trading off gain and phase}
The decentralized stability condition of Theorem~\ref{th.stability} imposes restrictions on the bus dynamics that change with $\alpha$. Notably, $\alpha$ can be understood as parameterizing a trade-off between gain and phase bounds on $h_n(j\omega_p)$ and, implicitly, $m_{p,n}(j\omega_p)$. While $\alpha$ can be chosen to accommodate different unit dynamics (e.g., IBR controls), we focus on the lower and upper limits of $\alpha$, which we refer to as the low-gain and passive case, respectively (see Fig.~\ref{fig:nyquist-cases}).  
%
\begin{figure}[b!!] 
    \centering
    \includegraphics[width=\linewidth]{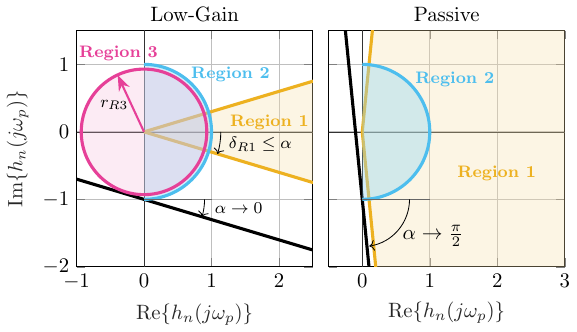}
    \caption{Simplified sufficient conditions can be constructed from Theorem~\ref{th.stability} for $h_n(j \omega_p)$ with low gain at high frequency (left) and passive $h_n(j \omega_p)$ (right) that correspond to  $\alpha \to 0$ (left) and $\alpha \to \frac{\pi}{2}$ (right).\label{fig:nyquist-cases}}
\end{figure}
In addition, we map the limits on the Nyquist plot of $h_n(j\omega_p)$, which includes line and unit dynamics, to sufficient conditions on the Bode plot of $m_{p,n}(j\omega_p)$ to obtain decentralized stability conditions on the unit dynamics that are parameterized in key network parameters.

To this end, we partition the admissible region of the Nyquist plot into three regions  (see Fig.~\ref{fig:nyquist-cases}).  Region 1 (R1) captures strictly dissipative bus dynamics, where $\angle h_n(j\omega_p) \leq \delta_{R1}$ and $0 \leq \delta_{R1} \leq \alpha$, that never cross the boundary imposed by Theorem~\ref{th.stability}. Region 2 (R2) is a semicircle that imposes the gain limit $| h_n(j\omega_p)| < 1$ and dissipativity condition $|\angle h_n(j \omega_p)| \leq \frac{\pi}{2}$, a relaxation of the strict dissipativity condition of R1.  Finally, Region 3 (R3) is a circle with radius $r_{R3} = \cos(\alpha)$ that encodes the gain bound $| h_n(j\omega_p)| < \cos(\alpha)$ but does not require dissipativity such that there are no restrictions on phase. It should be noted that as $\alpha \to \frac{\pi}{2}$, $r_{R3} \to 0$.  Thus, the transition from the low-gain to the passive case encodes a trade-off between the sizes of R1 and R3 and can be used to tailor the conditions of Theorem~\ref{th.stability} to different unit (e.g. IBR control) characteristics.

\subsection{Conditions on unit dynamic droop coefficients}
The regions in Fig.~\ref{fig:nyquist-cases} that restrict the bus dynamics \eqref{eqn:interconnection} can be mapped to the Bode plot of $m_{p,n}(j\omega_p)$ to clarify how the decentralized stability condition restricts the permissible unit dynamics, i.e., the dynamic droop coefficient. To this end, we align the regions in Fig.~\ref{fig:nyquist-cases} with different frequency ranges and, leveraging the bounds \eqref{eq:mubound} on the (normalized) line dynamics, map them to the Bode plot of $m_{p,n}(j\omega_p)$.

\subsubsection{High-frequency low-gain case ($\alpha \to 0$)}\label{subsec:lowgain}
The boundary between R1 and R2 for the low-gain case is selected as the lowest frequency $f_\alpha$ for which $| \angle \mu(j 2\pi f_{\alpha}) + \angle m_{p,n}(j 2\pi f_{\alpha})| = \alpha$. Moreover, R2 and R3 are separated by the lowest frequency $f^{\prime \prime}$ for which $|\angle \mu(j 2\pi f^{\prime \prime}) + \angle m_{p,n}(j 2\pi f^{\prime \prime})| = \tfrac{1}{2}\pi$. In other words, a phase bound is enforced for $f_p \leq f_\alpha$, the mixed magnitude and phase bounds of R2 are enforced for $f_\alpha \leq f_p \leq f^{\prime \prime}$, and a gain bound is enforced for $f_p \geq f^{\prime \prime}$ once dissipativity of $h_n(j\omega_p)$ is lost. To obtain bounds on the Bode plot of $m_{p,n}(j\omega_p)$, we further require the frequency $f_\epsilon$ given by \eqref{eq:mubound:feps} after which the gain of the line dynamics deviates from its steady-state value.

Bounding $h_n(j\omega_p)$ by R1 for $f_p \leq f_\alpha$ requires that
\begin{align}\label{eq:bus-dyn-R1}
  -\alpha \leq \angle \mu(j 2\pi f_p) + \angle m_p(j 2\pi f_p) \leq \alpha.
\end{align}
Using $-\delta_\epsilon = \angle \mu(j 2\pi f_\epsilon)$ this implies that $h_n(j\omega_p)$  is in R1 if
\begin{align*}
  -\alpha + \delta_\epsilon \leq \angle m_{p,n}(j 2\pi f_p) \leq \alpha.
\end{align*}
Bounding $h_n(j\omega_p)$ by R2 for $f_{\alpha} < f_p \leq f^{\prime \prime}$ requires that
\begin{subequations}\label{eq:bus-dyn-R2}
  \begin{align}
  \tfrac{\gamma_n}{\psi_n \omega_p} \big|\mu(j 2\pi f_p)\big| \big|m_{p,n}(j 2\pi f_p)\big| < 1, \label{eq:bus-dyn-mag-limit-R2}\\
  -\tfrac{1}{2}\pi \leq \angle \mu(j 2\pi f_p) + \angle m_{p,n}(j 2\pi f_p) \leq \tfrac{1}{2}\pi.
\end{align}
\end{subequations}
Assuming that $f^{\prime\prime} \leq f_\epsilon$, then \eqref{eq:mubound:feps} and $  -\delta_\epsilon \leq \angle \mu(j 2\pi f_p) \leq 0$ holds. Thus, restricting the bus dynamics to 
\begin{align*}
  |m_{p,n}(j 2\pi f_p)| &< \tfrac{\psi_n \omega_p}{\gamma_n (\mu_0 + \epsilon_l)},\\
  -\tfrac{1}{2}\pi + \delta_\epsilon \leq \angle m_{p,n}(j 2\pi f_p) &\leq \tfrac{1}{2}\pi,
\end{align*}
ensures that $h_n(j\omega_p)$ is in R2.
\noindent
Finally, bounding $h_n(j\omega_p)$ by R3, where $f_p > f^{\prime \prime}$ requires that
\begin{align}\label{eq:bus-dyn-limit-R3}
  \tfrac{\gamma_n}{\psi_n \omega_p} |\mu(j 2\pi f_p)| |m_{p,n}(j 2\pi f_p)| < \cos(\alpha).
\end{align}
Considering \eqref{eq:mubound}, the peak gain $|\mu(j2\pi f_0)| = \hat{\mu}$, and $|\mu(j 2\pi f_p)| \leq \hat{\mu} + \epsilon_l$ for all $f_p > f_{\epsilon}$, the unit is required to dampen the resonant peak from the line dynamics such that
\begin{subequations}\label{eq:device-mag-limit-R3}
  \begin{align}
  |m_{p,n}(j 2\pi f_p)| &\leq \frac{\cos(\alpha)\psi_n \omega_p}{\gamma_n (\mu_0 + \epsilon_l)}, \quad &f^{\prime \prime} &< f_p \leq f_{\epsilon}\\
  |m_{p,n}(j 2\pi f_p)| &\leq \frac{\cos(\alpha)\psi_n \omega_p}{\gamma_n (\hat{\mu} + \epsilon_l)}, \quad &f_{\epsilon} &< f_p. 
\end{align}
\end{subequations}
The restrictions imposed on the unit dynamic droop coefficient $m_{p,n}(j 2\pi f_p)$ in the low-gain case are illustrated in the Bode plot in Fig.~\ref{fig:device-low-gain}.  We note that the values of $\epsilon_l$ and $\alpha$ will impact the bounds and frequency ranges.
\begin{figure}[b] 
    \centering
    \includegraphics[width=\linewidth]{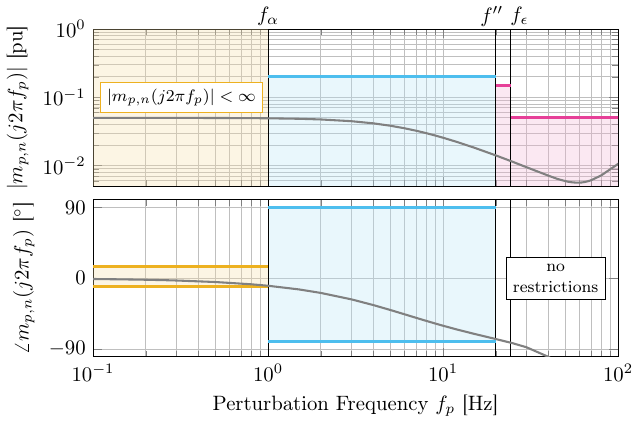}
    \caption{Illustration of the restrictions imposed on the dynamic droop coefficient corresponding to R1, R2, and R3 for the low-gain case, where $\alpha = 15^\circ$, $\epsilon_l = 20\%$, and $\delta_\epsilon = 4.5^\circ$.  A droop-controlled GFM IBR with inner-loops is shown to illustrate compliance with the low-gain conditions.
    \label{fig:device-low-gain}}
\end{figure}

\subsubsection{Passive case ($\alpha \to \tfrac{1}{2}\pi$)}
\begin{figure}[t] 
    \centering
    \includegraphics[width=\linewidth]{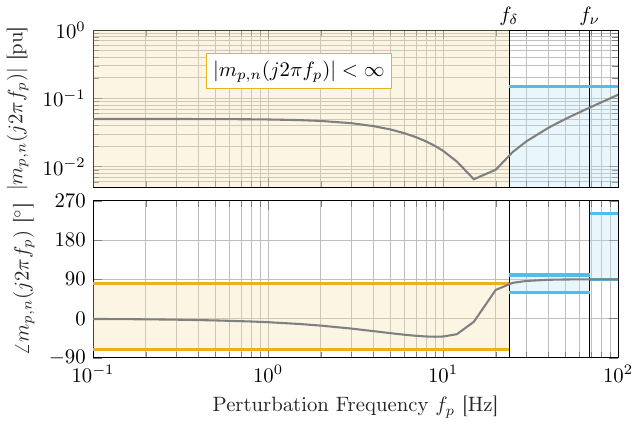}
    \caption{Illustration of the restrictions imposed on the dynamic droop coefficient corresponding to R1 and R2 for the passive case, with $\alpha = 90^\circ - \delta_\alpha$, $\delta_\alpha = 10^\circ$, $\delta_l = 4.5^\circ$, and $\nu_l = \delta_\alpha + \delta_l + 10^\circ$. A droop-controlled GFM IBR without inner-loops is shown to illustrate compliance with the passive conditions.    \label{fig:device-passive}}
\end{figure}
In the passive case, the unit compensates for the $180^\circ$ drop in phase at the resonance frequency of the normalized line dynamics $\mu(j 2\pi f_0)$ to meet the conditions of Theorem~\ref{th.stability}.  To this end, the boundary between R1 and R2 is placed at the lowest frequency $f_{\delta}$ for which $\angle \mu(j 2\pi f_\delta) = -\delta_l$ while R3 is not used. R2 is further partitioned using the lowest frequency $f_\nu$ such that $\angle \mu(j 2\pi f_\nu) = \nu_l -\pi$, i.e., the frequency at which the line dynamics approach $-180^\circ$ phase. The unit is expected to decrease its dynamic droop gain and increase its dynamic droop phase around this frequency to counter the destabilizing effect of the line dynamics.

For R1, where $f_p \leq f_\delta$, the line model satisfies $-\delta_l \leq \angle \mu(j 2\pi f_p) \leq 0$. Hence, the unit dynamics are restricted to
\begin{align}\label{eq:bus-dyn-R1-passive}
  -\alpha +\delta_l \leq \angle m_{p,n}(j 2\pi f_p) \leq \alpha,
\end{align}
such that $h_n(j\omega_p)$ remains in R1.
For R2, where $f_p > f_\delta$, the line model satisfies $|\mu(j 2\pi f_p)| \leq \hat{\mu} + \epsilon_l$ and the phase bounds (see \eqref{eq:mubound})
\begin{align*}
  -\pi + \nu_l \leq &\angle \mu(j 2\pi f_p) \leq -\delta_l,& \quad f_\delta &< f_p \leq f_\nu,\\
  -\pi  \leq &\angle \mu(j 2\pi f_p) < -\pi  + \nu_l,& \quad f_\nu &< f_p.
\end{align*}
Thus, if the dynamic droop coefficient satisfies
\begin{subequations}
  \begin{align}
      \!\!\!  |m_{p,n}(j 2\pi f_p)| &< \tfrac{\psi_n \omega_p}{\gamma_n (\hat{\mu} + \epsilon_l)}, &\; f_\delta &< f_p  \label{eq:device-mag-limit-R2}\\ 
  \!\!\!\tfrac{1}{2}\pi-\nu_l < \angle m_{p,n}(j 2\pi f_p) &\leq \tfrac{1}{2}\pi +\delta_l, &\; f_\delta &< f_p \leq f_\nu\\
\!\!\!\tfrac{1}{2}\pi < \angle m_{p,n}(j 2\pi f_p) &\leq \tfrac{3}{2}\pi -\nu_l, &\; f_\nu &< f_p,
\end{align}
\end{subequations}
then $h_n(j\omega_p)$ is in R2. The restrictions imposed on the unit dynamic droop coefficient $m_{p,n}(j 2\pi f_p)$ in the passive case are illustrated in the Bode plot in Fig.~\ref{fig:device-passive}.  We note that the values of $\delta_l$ and $\nu_l$ will impact the requirements.

\section{Performance Specifications\label{sec:performance}}
While the previous conditions guarantee small-signal stability, they do not imply good control performance or provision of ancillary services. In the context of this work, performance refers to the disturbance rejection of a unit at its point of interconnection and discussions about performance are focused on the impact of active power disturbances on the frequency of the ac voltage at the unit terminal.

Firstly, steady-state droop requires that the gain and phase is within a small tolerance $\epsilon_d\in \mathbb{R}_{>0}$ and $\delta_d$ (e.g., obtained from standards or grid-codes) of the desired droop constant $m_{p,0}$ up to some frequency $f_d$. In other words, for $f_p \leq f_d$, we require that
\begin{subequations}
\begin{align} \label{eq:perf-spec-1}
  m_{p,0} - \epsilon_{d} &\leq |m_{p,n}(j2\pi f_p)| \leq m_{p,0} + \epsilon_{d},\\
    -\delta_d &\leq \angle m_{p,n}(j2\pi f_p) \leq \delta_d.
\end{align}
\end{subequations}
Another performance specification that may be imposed is the transient damping specification
\begin{subequations}\label{eq:perf-spec-2}
  \begin{align} 
  |m_{p,n}(j2\pi f_p)| &\leq \overline{m}_{p,n}, \label{eq:perf-spec-2:gain} \\
    -\tfrac{1}{2}\pi \leq  \angle m_{p,n}(j2\pi f_p) &\leq \tfrac{1}{2}\pi,
\end{align}
\end{subequations}
which requires the unit to maintain a bounded gain from power oscillations to frequency oscillations beyond steady-state droop.

Finally, an inertia-like response can be specified by requiring that $|m_{p,n}(j2\pi f_p)|$ decreases ten times for a ten times frequency increase beyond a cut-off frequency $f_c$ (see Fig.~\ref{fig:perf-bode}). To illustrate the interpretation of this requirement, consider the GFM VSM control with virtual inertia constant $H$ and damping $D=m_{p,0}^{-1}$ modeled by the transfer function
\begin{align} \label{eq:VSM}
  g_{\text{VSM}}(s) = \frac{1}{2Hs+D} = \frac{1}{D}\frac{1}{\frac{2H}{D}s+1}
\end{align}
between per unit active power and per unit IBR frequency.  Notably, beyond the cut-off frequency $f_c = \pi D/H$, the gain of \eqref{eq:VSM} decreases ten times for a ten times frequency increase. 

For typical parameters, the performance specifications are more restrictive than the decentralized stability conditions for the low-gain case (see Fig.~\ref{fig:perf-bode}). For $f_p > f^{\prime\prime}$ the stability conditions already require decreasing gain, which reduces the amplitude of frequency oscillations, and hence no additional performance specifications are required.

\begin{figure}[t]
    \centering
    \includegraphics[width=3.5in]{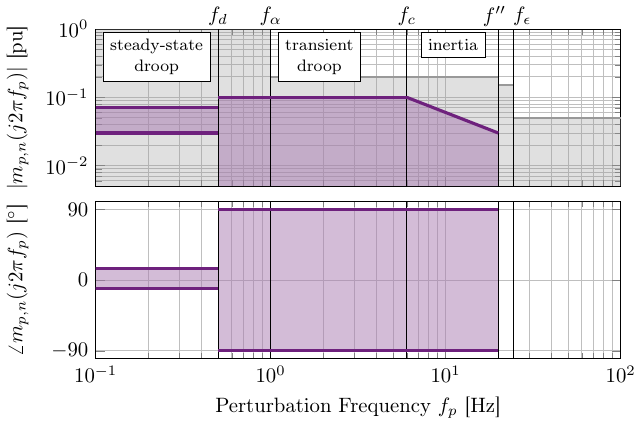}
    \caption{Performance specifications (purple) impose stricter requirements on the low-frequency unit dynamics than the low-gain stability condition (grey).
      \label{fig:perf-bode}}
\end{figure}

\section{Numerical Examples\label{sec:examples}}
In this section, we provide examples that illustrate use of the stability and performance specifications and the data-driven verification method to evaluate unit dynamics.

\subsection{Minimum network impedance for stability}

First, we consider IBRs that meet the performance specifications and illustrate that increasing network impedance increases the range of droop coefficients for which stability can be certified using the bounds proposed in Sec.~\ref{sec:stability}. To this end, we focus on the low-gain case, where $\alpha \to 0$, and note that the stability conditions in R1 given by \eqref{eq:bus-dyn-R1} and \eqref{eq:bus-dyn-R1-passive} do not depend on $\gamma_n$, making them independent of the network impedance. Next, we replace $\gamma_n$ with its upper bound $\bar{\gamma}$ in the stability conditions. It remains to investigate the impact of $\ell_{\min}$ in R2 and R3. Typically, in the low-gain case $f_\alpha \leq f_c  \leq f^{\prime\prime} \leq f_\epsilon$ and in the passive case $f_\delta \leq f_\epsilon$. Then, substituting \eqref{eq:perf-spec-2} into \eqref{eq:bus-dyn-mag-limit-R2}, using $\mu(j2\pi f_p) \leq \tfrac{1}{\rho^2+\omega_0^2}+\epsilon_l$, and using the definition of $\bar{\gamma}$, the R2 stability specification for $f_\alpha < f_p \leq f_c$ becomes
\begin{align}\label{eq:region21}
   \ell_{\min} > \frac{2e_{\max,n}}{\psi_n 2\pi f_\alpha} \left(\frac{1}{\rho^2+\omega_0^2}+\epsilon_l\right) \overline{m}_{p,n}.
\end{align}
Applying similar steps to \eqref{eq:bus-dyn-limit-R3}, the minimum line impedance for $f_c < f_p \leq f^{\prime\prime}$ is
\begin{align} \label{eq:region22}
  \ell_{\min} > \frac{1}{\cos(\alpha)}\frac{2e_{\max,n}}{\psi_n 2 \pi {f_c}}\left(\frac{1}{\rho^2+\omega_0^2} + \epsilon_l\right) \frac{\overline{m}_{p,n}}{\sqrt{(\frac{f^{\prime\prime}}{f_c})^2 + 1}}.
\end{align}
Finally, the bound on $|m_{p,n}(j2\pi f_p)|$ that defines R3 in \eqref{eq:device-mag-limit-R3} scales with $\bar{\gamma}^{-1}$ and increases linearly with $\ell_{\min}$.

Fig.~\ref{fig:contour1} illustrates that the minimum network inductance required to meet \eqref{eq:region21} increases with the maximum transient droop $\overline{m}_{p,n}$ and decreases as $\rho$ increases, i.e., as the network resistance increases. It should be noted that a larger threshold for $\epsilon_l$ will raise the lower limit on $\ell_{\min}$. To illustrate the trade-off between inertia, transient droop, and $\ell_{\min}$, for a GFM VSM we substitute $f_c = \pi D/H$ (see Sec.~\ref{sec:performance}) into \eqref{eq:region22} and express $f^{\prime\prime}$ as a function of $\rho$ by solving $|\mu(j2\pi f^{\prime\prime})| = \mu_0 + \epsilon_l$. The resulting minimum network inductance shown in Fig.~\ref{fig:contour2} increases with the maximum transient droop $\overline{m}_{p,n}$ and decreases as the inertia $H$ increases. Notably, minimum network inductance in the transient droop range $f_\alpha < f_p \leq f_c$ is significantly larger than that in the inertia range $f_c < f_p \leq f^{\prime\prime}$. This observation aligns with results for networks of homogeneous GFM IBRs which showed that the minimum stabilizing network inductance in many cases does not depend on the inertia characteristics but only on the transient droop characteristics~\cite{compensate}.
\begin{figure}[t]
  \centering
  \includegraphics[width=\columnwidth]{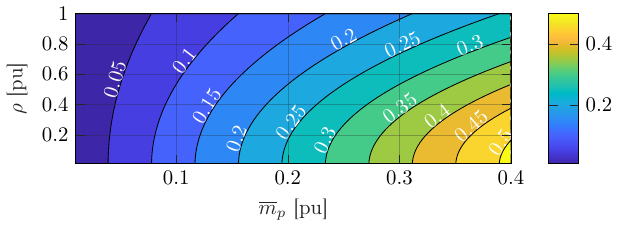}
  \caption{Minimum network inductance $\ell_{\min}$ in the frequency range of $f_\alpha < f_p \leq f_c$, resulting from the settings of $f_\alpha = 0.6$ Hz, $\psi = 1$, $e_{\max,n}=2$, $\epsilon_l = 0.1\mu_0$, and $\epsilon_d = 0.1m_{p,0}$.
    \label{fig:contour1}
  }
\end{figure}
\begin{figure}[t]
  \centering
  \includegraphics[width=\columnwidth]{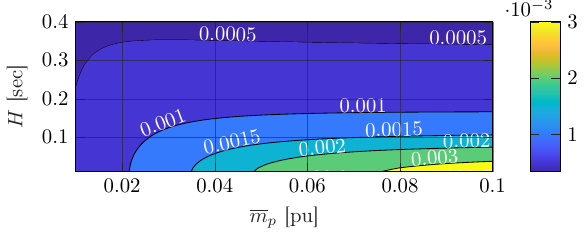}
  \caption{Minimum network inductance $\ell_{\min,n}$ for $f_c < f_p \leq f^{\prime\prime}$, where the unit is required to demonstrate an inertial response.  The results are based on $f_\alpha = 0.6$ Hz, $\psi = 1$, $e_{\max,n}=2$, $\epsilon_l = 0.1\mu_0$, $\epsilon_d = 0.1m_{p,0}$, and $\rho = 0.1$.
    \label{fig:contour2}
  }
\end{figure}
\subsection{Grid-forming vs. grid-following response}
The stability and performance conditions developed on the dynamic droop coefficient may also be used to delineate GFM and GFL dynamics. The most prevalent form of GFM control are droop and VSM control, which impose the linear relationship
\begin{align}        \label{eqn:gfm_c}
  \Delta \omega = -\frac{m_{p,0}}{\tau s + 1} \Delta p,
\end{align}
between the IBR active power injection and IBR frequency. Here $\tau \in \mathbb{R}_{>0}$ is a time constant used for filtering active power measurements or inertia emulation (i.e., the Bode magnitude plot roll off after $\tau^{-1}$) to reduce frequency oscillations.  These traits may be tuned to meet the local performance specifications outlined in the previous subsection. In contrast, common implementations of GFL control with $P\text{-}f$ droop utilize a SRF-PLL. The GFL IBR is controlled to track an active power reference obtained by $\Delta p =  -g_D(s)g_{PLL}(s) \Delta \omega$, where $g_{PLL}(s)= \frac{k_ps + k_i}{s^2 + k_ps + k_i}$ and $g_D(s)=\frac{D}{\tau_ds + 1}$. This results in 
%
    \begin{align} \label{eqn:gfl_c}
      \Delta \omega = -\frac{s^2 + k_ps + k_i}{k_ps + k_i} \frac{\tau_ds + 1}{D} \Delta p
    \end{align}
where $k_p \in \mathbb{R}_{>0}$ and $k_i \in \mathbb{R}_{>0}$ are the PLL gains, $\tau_d \in \mathbb{R}_{>0}$ is a filter time constant of a realizable differentiator for implementing frequency droop, and $D \in \mathbb{R}_{>0}$ is the damping constant. Note that the DC gain, i.e., steady-state droop gain, of the GFL response is given by $1/D$. However, \eqref{eqn:gfl_c} is improper and non-causal and the control gains cannot be tuned to satisfy the performance and stability specifications.

\begin{proposition}{\bf(SRF-PLL damping)}\label{prop.SRFPLL}
  Consider the SRF-PLL control \eqref{eqn:gfl_c}. For all coefficients $k_p$, $k_i$, and $\tau_d$, it holds that $|m_{p,n}(j 2\pi f_p)| \rightarrow \infty$ and $\angle m_{p,n}(j 2\pi f_p) \to \pi$ as $f_p \rightarrow \infty$.
\end{proposition}

\begin{proof}
  Let $m_{p,n}(s)=g_{\text{gfl}}(s)$. Then, the relative degree of $m_{p,n}(s)$ is two and it follows that $|m_{p,n}(j 2\pi f_p)| \rightarrow \infty$ and $\angle m_{p,n}(j 2\pi f_p) \to \pi$ as $f_p \rightarrow \infty$.   \hfill $\blacksquare$
\end{proof}
This contradicts the stability conditions \eqref{eq:device-mag-limit-R3} and \eqref{eq:device-mag-limit-R2} for the low-gain and passive case, respectively, and the inertia specification of the performance specifications. Due to PLL bandwidth limits, these violations will typically occur in the frequency range of interest. Moreover, modeling the full GFL IBR dynamics (see Fig.~\ref{fig:dynamics-results}) will typically result in more significant bound violations than the reduced-order model \eqref{eqn:gfl_c}.

Fig.~\ref{fig:gfm-vs-gfl} illustrates the dynamic droop coefficients for three IBR controls obtained in Sec.~\ref{subsec:example}.  It can be seen that both GFM IBRs meet all three performance specifications, while the GFL IBR only meets the steady-state and transient droop specifications. Notably, the GFL unit does not provide significant damping for high-frequency disturbances. Instead, the GFL response results in negative damping, since the phase is outside the range $[-90^\circ,90^\circ]$ for frequencies beyond $f^{\prime\prime}$. While the GFM IBR with inner loops also exhibits negative damping beyond approximately $f_\epsilon$, this is compensated by its low-gain (see Sec.~\ref{subsec:lowgain}). In contrast, the GFM IBR without inner loops does not exhibit negative damping but does result in larger frequency deviations beyond approximately  $f_\epsilon$.
\begin{figure}[b!]
    \centering
    \includegraphics[width=\columnwidth]{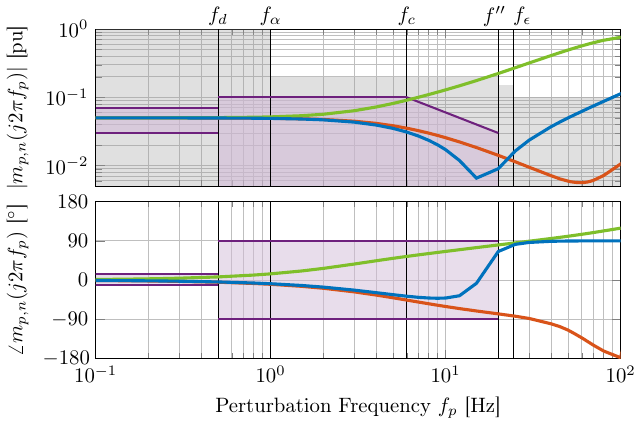}
    \caption{Bode plot of $m_{p,n}(j 2\pi f_p)$ for GFM with inner-loops (\drawline[orange1, line width=1.3pt]), 
    GFM with inner-loops (\drawline[blue, line width=1.3pt]), and 
    GFL (\drawline[green, line width=1.3pt]) for  $\alpha=15^\circ$ and $\delta_\epsilon = 4.5^\circ$.
    \label{fig:gfm-vs-gfl}
    }
\end{figure}

\subsection{Impact of line model on certifying stability}
Finally, we briefly illustrate how our framework can be used to explain prior results that show that considering network circuit dynamics can significantly impact stability of networks of GFM and GFL IBRs~\cite{GCB+2019,MSA+2021}.

\begin{figure}[t]
  \centering
  \includegraphics[width=3.5in]{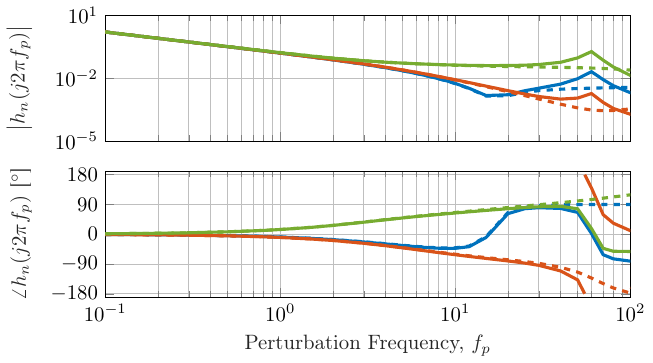}
  \caption{Bode plot of $h_n(j 2\pi f_p)$ for a GFM IBR with inner-loops (\drawline[orange1, line width=1.3pt]), GFM IBR without inner-loops (\drawline[blue, line width=1.3pt]), and   GFL IBR (\drawline[green, line width=1.3pt]), with (solid) and without (dashed) network circuit dynamics.
  \label{fig:inter_bode}}
\end{figure}

Fig.~\ref{fig:inter_bode} shows the Bode plot of $h_n(j 2\pi f_p)$ for three IBR controls obtained from the data in Sec.~\ref{subsec:example} with and without network circuit dynamics (i.e., replacing $\mu(j 2\pi f_p)$ with $\mu_0$). It can be seen that $h_n(j 2\pi f_p)$ for the GFL IBR is passive with the quasi-steady-state line model, but is passive with the dynamic line model. This supports the observation made in \cite{MSA+2021} that transmission line dynamics may improve the stability margins of GFL IBRs. On the other hand, the Nyquist plot of the GFL IBR response presented in Fig.~\ref{fig:stability-nyquist-1} highlights that the GFL response with line dynamics has smaller stability and robustness margin than that of both GFM IBRs. Moreover, $h_n(j 2\pi f_p)$ for the GFM IBR exhibits both increased gain and phase significantly outside $[-90^\circ,90^\circ]$ when considering line dynamics. This precludes meeting the stability conditions for the passive case and makes meeting the low-gain stability conditions challenging for networks with low impedance, which are known to pose stability challenges to GFM IBRs~\cite{GCB+2019}. Meanwhile, the GFM IBR without inner loops remains dissipative both with and without line dynamics and only exhibits an increase in gain that does not appear critical in light of the stability conditions for the passive case.

\section{Conclusion\label{sec:conclusions}}
In this work, we developed a modeling and analysis framework for certifying small-signal frequency stability of a network of heterogeneous units (e.g., IBRs and SGs) using decentralized stability conditions on data-driven input-output models that do not require knowledge of internal hardware or control structure. Moreover, we used the data-driven input-output models to evaluate the unit-level frequency control performance. We extended the notion of a steady-state droop coefficients to dynamic droop coefficients in order to evaluate frequency stability for networks of heterogeneous units considering network circuit dynamics.  Furthermore, the dynamic droop coefficients allow for formalizing grid support functions such as inertia emulation into verifiable requirements on the unit dynamics. Our decentralized stability condition and performance specifications were formulated as bounds on the Nyquist and Bode plots and apply to GFM IBRs, GFL IBRs, and SGs.  Developing conditions for voltage stability and recovering unit dynamics from hardware experiments are seen as important topics for future work.

\bibliographystyle{IEEEtran}
\bibliography{IEEEabrv,bibliography}

\vfill

\end{document}